\documentclass[submission,copyright,creativecommons]{eptcs}

\providecommand{\event}{GandALF 2021} 
\usepackage{underscore}           

\usepackage[utf8]{inputenc} 
\usepackage{amsmath,amssymb} 
\usepackage{amsthm} 
\usepackage{mathtools} 
\usepackage{stmaryrd} 
\usepackage{thm-restate} 
\usepackage{microtype} 
\usepackage{cleveref}

\usepackage[caption=false,font=footnotesize]{subfig}
\usepackage{tikz}
\usetikzlibrary{matrix,arrows,calc,positioning,backgrounds,patterns,decorations.pathmorphing}

\newtheorem{theorem}{Theorem}
\newtheorem{proposition}[theorem]{Proposition}

\newtheorem{corollary}[theorem]{Corollary}

\theoremstyle{definition}
\newtheorem{definition}[theorem]{Definition}
\newtheorem{example}[theorem]{Example}

\newcommand{\cornerhere}{\hfill$\lrcorner$\gdef\ExampleEndMarker{}}
\newenvironment{Example}{\gdef\ExampleEndMarker{\hfill$\lrcorner$}\example}{\ExampleEndMarker\endexample}

\newcommand{\E}{\exists}
\newcommand{\A}{\forall}
\renewcommand{\phi}{\varphi}
\renewcommand{\emptyset}{\varnothing}
\newcommand*{\ext}[1]{[\![ #1 ]\!]}
\DeclareMathOperator{\lfp}{\mathbf{lfp}}
\DeclareMathOperator{\gfp}{\mathbf{gfp}}

\newcommand{\dcup}{\dot\cup}
\newcommand*{\tup}[1]{\mathbf{#1}}
\newcommand{\ta}{\tup a}
\newcommand{\tb}{\tup b}
\newcommand{\tx}{\tup x}
\newcommand{\ty}{\tup y}
\newcommand{\co}{\colon}

\renewcommand{\AA}{{\mathfrak A}}
\newcommand{\N}{{\mathbb N}} 
\newcommand{\K}{{\mathbb K}} 
\newcommand{\Gg}{\mathsf{G}} 
\newcommand{\Ss}{\mathsf{S}} 
\newcommand{\Tt}{{\mathsf T}} 

\DeclareMathOperator{\Strat}{\mathrm{Strat}}
\DeclareMathOperator{\WinStrat}{\mathrm{WinStrat}}
\DeclareMathOperator{\Lit}{\mathrm{Lit}}

\newcommand{\Bool}{\mathbb{B}}
\newcommand{\Nat}{\mathbb{N}}
\newcommand{\Sinf}{{\mathbb S}^{\infty}}
\newcommand{\Trop}{\mathbb{T}}
\newcommand{\Vit}{\mathbb{V}}
\newcommand{\PosBool}{\mathsf{PosBool}}

\renewcommand{\bar}{\overline}
\newcommand*{\ps}[1]{[\![ #1 ]\!]}
\newcommand*{\nn}[1]{\bar{#1}}
\newcommand{\nnX}{\nn{X}}
\newcommand{\nnx}{\nn{x}}

\newcommand{\Inf}{\bigsqcap}
\newcommand{\Sup}{\bigsqcup}
\newcommand{\Bone}{\mathbf{1}}
\newcommand{\Bzero}{\mathbf{0}}
\newcommand{\bcdot}{\boldsymbol\cdot} 
\newcommand{\FormulaWin}{\mathsf{win}_0}

\newcommand{\absorb}{\succeq}
\newcommand{\absorbneq}{\succ}

\newcommand{\var}{\mathsf{var}}

\newcommand{\pitrack}{\pi_{\text{\sffamily strat}}}
\newcommand{\pirev}{\pi_{\text{\sffamily rep}}}

\newcommand{\prefix}{\sqsubseteq}
\newcommand{\ecount}[2]{\#_{#2}(#1)} 
\newcommand{\ep}[1]{\#_E(#1)} 

\tikzset{
    arr/.style={draw,->,>=stealth',shorten <=2pt,shorten >=2pt,every node/.style={auto,inner sep=2pt,font=\scriptsize}},
    gamenode/.style={draw,inner sep=2pt,minimum size=.4cm},
    p0/.style={gamenode,circle},
    p1/.style={gamenode,rectangle},
    F/.style={thick, pattern=north east lines},
    dot/.style={circle,draw,fill,black,minimum size=3pt,inner sep=0pt},
    marker/.style={draw=none,inner sep=0pt,overlay},
    short/.style={ shorten >=#1, shorten <=#1 }
}

\title{Semiring Provenance for Büchi Games:\\Strategy Analysis with Absorptive Polynomials}

\author{Erich Grädel
\institute{RWTH Aachen University\\Aachen, Germany}
\email{graedel@logic.rwth-aachen.de}
\and
Niels Lücking
\institute{RWTH Aachen University \\ Aachen, Germany}
\email{niels.luecking@rwth-aachen.de}
\and
Matthias Naaf
\institute{RWTH Aachen University \\ Aachen, Germany}
\email{naaf@logic.rwth-aachen.de}
}
\def\titlerunning{Strategy Analysis in Büchi Games by Semiring Valuations}
\def\authorrunning{E. Grädel, N. Lücking \& M. Naaf}

\begin{document}
\maketitle

\begin{abstract}
This paper presents a case study for the application of semiring semantics for fixed-point formulae
to the analysis of strategies in Büchi games.
Semiring semantics generalizes the classical Boolean semantics by permitting multiple truth values from certain semirings.
Evaluating the fixed-point formula that defines the winning region in a given game in an appropriate
semiring of polynomials provides not only the Boolean information on \emph{who} wins, but also tells us \emph{how} they win and \emph{which strategies} they might use. This is well-understood for reachability games, where the winning region is definable as a least fixed point.
The case of Büchi games is of special interest, not only due to their practical importance, but also because it is the simplest case where
the fixed-point definition involves a genuine alternation of a greatest and a least fixed point.

We show that, in a precise sense, semiring semantics provide information about all \emph{absorption-dominant} strategies -- strategies that win with minimal effort,
and we discuss how these relate to positional and the more general \emph{persistent} strategies.
This information enables applications such as game synthesis or determining minimal modifications to the game needed to change its outcome.\\

\noindent
Due to space reasons, several proofs have been omitted and can be found
in the \href{https://arxiv.org/abs/2106.12892}{full version} \cite{GraedelLucNaa21}.
\end{abstract}


\section{Introduction}

Two-player games on finite graphs which admit infinite plays are of fundamental importance in many areas of logic and computer science, especially in the 
formal analysis of reactive systems, where they model the non-terminating interaction between a system and its environment.
In such a game, the \emph{objective} or \emph{winning condition} of the player who represents the system specifies
the desired set of behaviours of the system. The most basic classes of such
objectives are \emph{reachability} and \emph{safety} objectives defined by a set of states (positions) that the player should reach, or avoid.
We can assume, without loss of generality, that even though infinite plays are possible in a game with reachability or safety objectives, 
they are all won by the same player.

Games with genuine and non-trivial winning conditions for infinite plays are harder to analyse; they include games with arbitrary
$\omega$-regular objectives, such as liveness, Muller, Streett-Rabin, or parity objectives, and many others. 
The goal of this paper is to provide a case study of a recent method for strategy analysis, based on semiring semantics,
and we would like to explore its potential for providing detailed information about strategies in genuinely infinite games.  
One of the simplest class of games with a non-trivial winning condition for the infinite plays are games with 
Büchi objectives, which require that a specific target set $F$  of states is reached infinitely often during the play (see e.g. \cite{GraedelThoWil02} for background). 
Büchi games, as well as some of their straightforward generalisations, have many applications in formal methods,
and efficient algorithms for solving them have been studied thoroughly (see e.g. \cite{ChatterjeeDvoHenLoi16,ChatterjeeHen14,ChatterjeeHenPit08}).
They are also of interest from the points of view of topology and logic, because they are among the simplest games where the set of winning plays is neither open nor closed, and where logical definition of the winning region requires a genuine alternation of a greatest and a least fixed point (see Sect.~\ref{sec:TheFormula}).
 
Strategies in infinite games can be very complicated because, in principle, they may depend on the entire history of a play.
Thus, there exist uncountably many different strategies, even on a finite game graph.
Fortunately, in many cases and in particular for Büchi games, simple strategies are sufficient to win.
A fundamental result in this context is the positional determinacy of parity games (of which Büchi games are a special case),
saying that from each position, one of the two players has a \emph{positional winning strategy}, i.e.
a winning strategy that only depends on the current position and not on the history of the play.   
A positional strategy can be viewed as a subgraph of the game graph, and can therefore be represented in
a compact way. As a consequence, the algorithmic analysis of Büchi games has concentrated almost exclusively on the positional strategies. Here we extend this point of view somewhat and take also other kinds of simple strategies into account.  
Specifically, we are interested in \emph{absorption-dominant} winning strategies \cite{GraedelTan20} which are strategies without redundant
moves; this means that taking away anything, in the sense of demanding that some specific move is played less often,
makes the strategy non-winning. 
Another way to distinguish positional strategies from absorption-dominant ones concerns their minimisation properties:
while positional strategies minimize the \emph{set} of moves that they use, absorption-dominant strategies
take multiplicities into account and minimize the \emph{multiset} of moves.  
A further interesting class are the \emph{persistent} strategies \cite{MarcinkowskiTruderung02}, which are positional in each individual play
but not necessarily across distinct plays.
We shall study the relationship between these different classes of simple strategies, and prove that
every positional strategy is absorption-dominant and every absorption-dominant strategy is persistent,
and that these inclusions are strict.

The specific method for strategy analysis that we want to apply to Büchi games in this paper is
based on the logical definability of the winning positions by 
a formula in the fixed-point logic LFP, and on the semiring semantics for LFP
developed in \cite{DannertGraNaaTan21}.
In the classical Boolean semantics, a model $\AA$ of a formula $\phi$
assigns to each (instantiated) literal a Boolean value.
$\K$-interpretations $\pi$, for a suitable semiring $\K$, generalize this
by assigning to each such literal a semiring value from $\K$.
We then interpret $0$ as \emph{false} and all other semiring values as \emph{nuances of true}
that provide additional information, depending on the semiring:
For example, the Boolean semiring $\Bool = (\{0,1\}, \lor, \land, 0, 1)$ corresponds to Boolean semantics, the Viterbi-semiring $\Vit = ([0,1], \max, \cdot, 0, 1)$ can model \emph{confidence} scores, 
the tropical semiring $\Trop= (\mathbb{R}_{+}^{\infty},\min,+,\infty,0)$
is used for cost analysis, and min-max-semirings $(A, \max, \min, a, b)$ for a totally ordered set $(A,<)$ can model different access levels.
Most importantly, semirings of polynomials, such as $\N[X]$, allow us to \emph{track} certain literals by mapping them to different indeterminates. The overall value of the formula is then a polynomial that describes precisely what combinations of literals prove the truth of the formula.
Semiring semantics has been studied for various logics \cite{BourgauxOzaPenPre20, DannertGra19, DannertGra20,DannertGraNaaTan21,GraedelTan17},
following the successful development of semiring provenance in database theory and related fields (see e.g.\cite{DeutchMilRoyTan14,GeertsPog10,GreenKarTan07,GreenTan17,OzakiPen18,RaghothamanMenZhaNaiSch20,Senellart17}).
While semiring provenance analysis for database queries had largely been  confined
to positive query languages such as conjunctive queries, positive relational algebra, and Datalog,
the generalisation to logics such as first-order logic FO and least fixed-point logic LFP -- featuring full negation and unrestricted interaction between least and greatest fixed points -- poses non-trivial mathematical  challenges and requires  
new algebraic constructions. Specifically, it has turned out that appropriate semirings for LFP should be
absorptive and fully continuous. Fortunately, this is the case for most of the important application semirings such
as $\Vit, \Trop$ or min-max-semirings, but not for the natural semiring $\N$, or the general provenance semirings
of polynomials or formal power series, $\N[X]$ and $\N^\infty \ps X$. 
Instead, we rely on semirings $\Sinf[X]$ of \emph{generalized absorptive polynomials},
which we explain in Sect.~\ref{sect:semirings}, and which are
the \emph{universal} absorptive, fully-continuous semirings, in the sense that 
every mapping $h \colon X \to \K$ into an absorptive, fully-continuous semiring $\K$
uniquely extends to a fully-continuous semiring homomorphism $h \colon \Sinf[X] \to \K$,
see Theorem~\ref{universality}.  
From valuations of fixed-point formulae in such semirings we thus can 
derive detailed insights into why the formula holds -- and by applying this to
the fixed-point definition of winning positions in Büchi games we obtain compact descriptions of 
winning strategies, in particular of all positional strategies and all absorption-dominant ones.

After an analysis of simple winning strategies in Büchi games, and a short introduction to semiring semantics for fixed-point logic, we shall study the semiring valuations of the particular LFP-formula $\FormulaWin(x)$
that defines the winning region for Player~0 in Büchi games. Given that the objective of Player~0 is to ensure 
that the play hits the target set $F$
infinitely often, we may informally describe their winning region  as the \emph{largest} set
$Y$ of positions from which they can enforce a (further) visit to $Y\cap F$
after $k\geq 1$ moves.  On the other side the set of positions from which
Player~0 can enforce a visit to a target set is the \emph{smallest} set 
of positions that either are already in the target set, or from which Player~0 can 
enforce the play to come closer to it. Thus, the winning region of Player~0 
can be described as a greatest fixed point inside of which there is a least fixed point, and it is well-known that 
this fixed-point alternation in the treatment of Büchi objectives cannot be avoided, see e.g. \cite{BradfieldWal18}.

The theoretical underpinning of our method is a Sum-of-Strategies-Theorem, saying that for any position $v$ in a Büchi game, the
valuation of the LFP-formula $\FormulaWin(v)$ in an absorptive, fully-continuous semiring
coincides with the sum of the valuations of all absorption-dominant winning strategies from $v$.
The proof is somewhat involved and requires non-trivial machinery; details are given in the full version of this paper
\cite{GraedelLucNaa21}.
Besides being of theoretical interest, this result allows to study a number of interesting questions
concerning the available winning strategies in a given Büchi game:

\smallskip\noindent\textbf{Strategy tracking.}
Introducing indeterminates for all edges in a fixed Büchi game $\Gg$,
the semiring value $\pitrack \ext {\FormulaWin(v)}$ for a position $v$
is a polynomial whose monomials are concise descriptions of all absorption-dominant strategies.
From these monomials we can derive whether Player 0 wins from $v$ (if there are any monomials)
and which edges are used by each absorption-dominant strategy, and how often they appear in the strategy tree.
In particular, we can immediately identify and count positional strategies from the polynomial.
Going further, we can answer questions such as: can Player~0 still win if we remove edge $e$, or several edges at once?
Can they still win if edge $e$ may only be used finitely often in each play?

\smallskip\noindent\textbf{Repairing a game.} Instead of analysing strategies in a fixed game, we may also reason about modifications or synthesis of (parts of) the game.
For example, assuming Player~0 loses from $v$, what are minimal modifications to the game that would let Player~0 win from $v$?
To answer such questions we have to take into account also negative information (i.e., absent edges in the graph),
so as to find a minimal repair consisting of both moves to delete and moves to add.
Algebraically, this requires to extend our semirings by dual-indeterminates, which leads to quotient semirings
$\Sinf[X,\bar X]$ by a construction that has been used before in \cite{GraedelTan17,XuZhaAlaTan18,DannertGraNaaTan21} to deal with
semiring semantics for negation.
We illustrate with the example of minimal repairs that we can indeed derive the desired information from valuations in such semirings.

\section{Büchi Games and Strategies}

A Büchi game is given by a tuple $\Gg = (V, V_0, V_1, E, F)$ where
$V$ is a set of positions (here assumed to be finite), with a disjoint decomposition
$V=V_0 \dcup V_1$ into positions of Player~0 and positions of Player~1.
The relation $E\subseteq V\times V$ specifies the possible moves, and the target set
$F\subseteq V$ describes the winning condition.
We denote the set of immediate successors of a position $v$ by 
$vE:=\{w \mid vw\in E\}$ and require that $vE\neq\emptyset$ for all $v$.
A play from an initial position $v_0$ is an infinite path $v_0v_1v_2\dots$
through $\Gg$ where the successor $v_{i+1}\in v_iE$ is chosen by Player~0
if $v_i\in V_0$ and by Player~1 if $v_1\in V_1$.
A play $v_0v_1v_2\dots$ is won by Player~0 if $v_i\in F$ for infinitely many $i<\omega$,
otherwise it is won by Player~1. 
The winning region of Player~$\sigma$ is the set
of those positions $v\in V$ such that Player~$\sigma$ has a winning strategy
from $v$, i.e. a strategy that guarantees them a win, no matter what the opponent does.

A strategy for Player~$\sigma$ in $\Gg = (V, V_0, V_1, E, F)$ can be represented in different ways,
for instance as a function $f \co V^* V_\sigma \to V$ that assigns a next position to each partial play ending in a position of Player $\sigma$,
or simply $f \co V_\sigma \to V$ if the strategy is positional.
Here we follow an alternative approach and represent strategies as trees,
comprised of all plays that are consistent with the strategy (see, e.g., \cite{GraedelTan20}).
For simplicity, we only consider strategies of Player~0, so unless mentioned otherwise, ``strategy'' always refers to a strategy for Player~0.

\begin{definition}
Given a Büchi game $\Gg = (V, V_0, V_1, E, F)$, the \emph{tree unraveling} from $v_0$ is the tree $\Tt(\Gg, v_0)$ whose nodes are all finite paths $\rho$ from $v_0$ in $\Gg$ and whose edges are $\rho \to \rho v$ for $v \in V$.
We often write a node of $\Tt(\Gg, v_0)$ as $\rho v$ to indicate a finite path ending in $v \in V$. The length of $\rho$ is denoted by $|\rho|$ and we write $\rho \prefix \rho'$ if $\rho$ is a (not necessarily strict) prefix of $\rho'$.
\end{definition}

Strategies can then be defined as subtrees of the tree unraveling, which allows for a more visual way to reason about strategies.
An important detail is that the strategy tree only contains positions (and thus choices for these positions) that 
are reachable when following the strategy.
Moreover, we only consider finite Büchi games and hence the tree unraveling and all strategies are finitely branching. 

\begin{definition}
\label{def:strategyAsTree}
A \emph{strategy} $\Ss$ (of Player~0) from $v_0$ in $\Gg$ is a subtree of $\Tt(\Gg, v_0)$ induced by a node set $W$ satisfying the following conditions:\\
-\quad If $\rho v \in W$, then also $\rho \in W$ (prefix closure).\\
-\quad If $\rho v \in W$ and $v \in V_0$, then there is a unique $v' \in vE$ with $\rho v v' \in W$ (unique choice).\\
-\quad If $\rho v \in W$ and $v \in V_1$, then $\rho v v' \in W$ for all $v' \in vE$ (all moves of the opponent).\\
\noindent The strategy is winning if all plays contained in $\Ss$ are winning.

We commonly write $\rho \in \Ss$ instead of $\rho \in W$,
and we often refer to paths of the form $\rho v \in \Ss$ as \emph{occurrences of $v$} in $\Ss$.
When we depict strategies graphically, we represent finite paths $\rho v$ just by their last position $v$ to ease readability (notice that in the tree unravelling 
$\rho$ can be reconstructed from $v$ by following the path to the root).
See \Cref{fig:RunningStrategy} for an example.
For $v \in V_0$, we further write $\Ss(\rho v) = w$ if $\rho v w$ is the (unique) successor of $\rho v$ in $\Ss$.
If $\Ss$ is positional, we may also write $\Ss(v)$ to denote the unique successor of $v$ chosen by $\Ss$.
We write $\Strat_\Gg(v)$ and $\WinStrat_\Gg(v)$ to denote the set of all (winning) strategies of Player~0 from position $v \in \Gg$, and we drop $\Gg$ if the game is clear from the context.
\end{definition}

\begin{figure}[b]
\centering
\newcommand{\enode}[4][]{\draw [arr] (#2) edge node [#1] {$#4$} (#3);}
\subfloat[Rectangular nodes belong to Player~1, round nodes to Player~0, dashed nodes are in $F$.]{
\centering
\begin{tikzpicture}[node distance=1.5cm,framed,baseline]
\node [p1,label={left:$v$}] (0) {};
\node [p1, right of=0, yshift=.7cm] (1) {};
\node [p0, F, right of=1, yshift=-.7cm, label={below:$v'$}] (2) {};
\node [p0, right of=2, yshift=.7cm] (3) {};
\node [p1, F, right of=2, yshift=-.7cm, label={below:$u$}] (4) {};
\node [p0, F, below of=2,label={above:$w$}] (5) {};
\node [p1, left of=5] (6) {};
\draw [arr]
    (0) edge node {$a$} (1) (0) edge node {$c$} (2) (1) edge node {$b$} (2)
    (2) edge [bend left=20pt] node {$d$} (0)
    (2) edge node {$e$} (3)
    (3) edge [loop right] node {$i$} (3)
    (2) edge node [below left] {$f$} (4)
    (3) edge node {$h$} (4)
    (4) edge [loop right] node {$g$} (4)
    (4) edge node {$k$} (5)
    (5) edge [loop below] node {\strut$m$} (5)
    (5) edge node {$n$} (6)
    (6) edge [loop below] node {\strut$p$} (6)
    (6) edge [bend left] node {$q$} (0)
    ;
\end{tikzpicture}
\label{fig:RunningGame}
}
\hfill
\subfloat[Depiction of an infinite strategy tree of a winning strategy for Player~0 from position $v$.]{
\centering
\begin{tikzpicture}[node distance=1.3cm,every node/.style={scale=0.7},framed,baseline]
\node [p1,label={left:$v$}] (0) {};
\node [p1,right of=0,yshift=.8cm] (1) {};
\node [p0,F,right of=1,label={below:$v'$}] (2) {};
\node [p0,right of=2] (3) {};
\node [p1,F,right of=3] (4) {};
\node [marker,right of=4] (4a) {};
\node [marker,below of=4,xshift=.5cm] (4b) {};
\node [p0,F,right of=0,yshift=-.8cm,label={below:$v'$}] (a) {};
\node [p1,F,right of=a] (b) {};
\node [p1,F,below of=b,xshift=.5cm] (c) {};
\node [p1,F,below of=c,xshift=.5cm] (d) {};
\node [marker,right of=d] (d1) {};
\node [marker,below of=d,xshift=.5cm] (d2) {};
\node [p0,F,right of=b] (b1) {};
\node [p0,F,right of=b1] (b2) {};
\node [p0,F,right of=b2] (b3) {};
\node [marker,right of=b3] (b4) {};
\node [p0,F,right of=c] (c1) {};
\node [p0,F,right of=c1] (c2) {};
\node [p0,F,right of=c2] (c3) {};
\node [marker,right of=c3] (c4) {};
\enode 0 1 a
\enode 1 2 b
\enode 2 3 e
\enode 3 4 h
\enode[swap] 0 a c
\enode a b f
\enode[left] b c g
\enode[left] c d g
\enode b {b1} k
\enode {b1} {b2} m
\enode {b2} {b3} m
\enode c {c1} k
\enode {c1} {c2} m
\enode {c2} {c3} m
%
\draw [densely dotted,shorten <=3pt,shorten >=10pt] (4) edge (4a) (4) edge (4b);
\begin{scope}
\draw [densely dotted,shorten <=3pt,shorten >=10pt] (b3) edge (b4);
\clip (c3.north west) rectangle ($(c3.south east)+(.4cm,-.4cm)$);
\draw [densely dotted,shorten <=3pt,shorten >=10pt] (c3) edge (c4);
\end{scope}
\begin{scope}
\clip (d.north west) rectangle ($(d.south east)+(.4cm,-.4cm)$);
\draw [densely dotted,shorten <=3pt,shorten >=10pt] (d) edge (d1) (d) edge (d2);
\end{scope}
\end{tikzpicture}
\label{fig:RunningStrategy}
}
\caption{Running example of a Büchi game and a winning strategy.}
\label{fig:Running}
\end{figure}
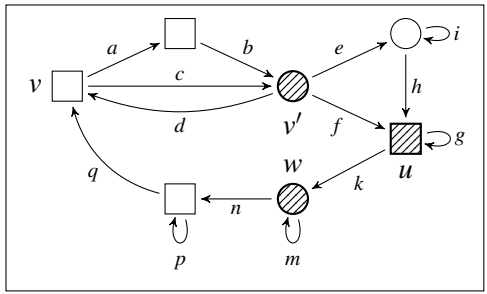
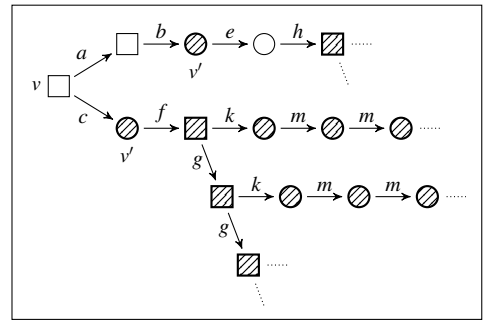

\begin{Example}
An example of a Büchi game is depicted in \Cref{fig:Running}.
Player~0 has essentially three different positional winning strategies from $v$,
by either choosing edge $d$, or edges $e,h,m$ or $f,m$.
Notice that for the first strategy, we did not specify moves for all positions in $V_0$ as these positions cannot be reached when edge $d$ is played; this is the main reason why we represent strategies as trees.
\Cref{fig:RunningStrategy} depicts such a tree representation of a strategy.
This strategy is a typical example of a winning strategy that is not positional, but still minimal if we take edge multiplicities into account.
\end{Example}

\section{Strategies with Minimal Effort}

\newcommand{\quoteauthor}{Author, Year}
\newenvironment{aquote}[1]{\renewcommand{\quoteauthor}{#1}\begin{quote}\itshape}{\hfill\textsc{--- \quoteauthor}\end{quote}}

\begin{aquote}{Antoine de Saint-Exupéry}
Perfection is achieved, not when there is nothing more to add, but when there is nothing left to take away
\end{aquote}

\noindent
As a measure for the complexity or effort of a strategy,
we consider the set of edges a strategy $\Ss$ uses and how often each of these edges appears in the strategy tree.
Under this measure, the simplest strategies are the ones that do not play redundant edges
-- hence no moves are left to take away.

\begin{definition}
Given an edge $e = vw \in E$ in a Büchi game $\Gg$ and a strategy $\Ss$ in $\Gg$,
we denote by $\ecount \Ss e = | \{ \rho v \in \Ss \mid \rho v \to \rho v w \text{ is an edge in } \Ss \}| \in \N \cup \{ \infty \}$ the number of times (possibly infinite) the edge $e$ occurs in $\Ss$.
With each strategy $\Ss$ we associate its \emph{edge profile}, the vector $\ep \Ss = (\ecount \Ss e)_{e \in E}$.
\end{definition}

\begin{Example}\label{ex:redundantMove}
Consider the following Büchi game:
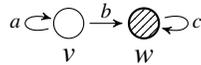
\begin{center}
\begin{tikzpicture}[baseline]
\node [p0,label={below:$v$},anchor=base,yshift=.1cm] (0) {};
\node [p0,F,label={below:$w$},right of=0] (1) {};
\draw [arr]
    (0) edge [loop left] node {$a$} (0)
    (1) edge [loop right] node {$c$} (1)
    (0) edge node {$b$} (1);
\end{tikzpicture}
\end{center}

Player~0 wins by first looping $n$ times at position $v$ (for any fixed $n \in \N$) and then moving to $w$, corresponding to the edge profile $(n,1,\infty)$.
Clearly, looping at $v$ is a redundant move, so we consider the strategy with $n=0$ as the simplest strategy (that wins with the least effort).
\end{Example}

To formalize the intuition of redundant moves, we define an order $\absorb$ on strategies called \textit{absorption}.
This is defined in such a way that the $\absorb$-maximal strategies are the simplest ones
that avoid redundant moves whenever possible.

\begin{definition}
Let $\Ss_1,\Ss_2$ be two strategies in a Büchi game $\Gg = (V,V_0,V_1,E,F)$.
We say that $\Ss_1$ \emph{absorbs} $\Ss_2$, denoted $\Ss_1 \absorb \Ss_2$,
if $\ecount {\Ss_1} e \le \ecount {\Ss_2} e$ for all edges $e \in E$.
If additionally $\ecount {\Ss_1} e < \ecount {\Ss_2} e$ for some $e \in E$,
we say that $\Ss_1$ \emph{strictly absorbs} $\Ss_2$, denoted $\Ss_1 \absorbneq \Ss_2$.
They are \emph{absorption-equivalent}, denoted $\Ss_1 \equiv \Ss_2$, if both $\Ss_1 \absorb \Ss_2$ and $\Ss_2 \absorb \Ss_1$.
A strategy $\Ss \in \Strat(v)$ is \emph{absorption-dominant from position $v$}, if there is no strategy $\Ss' \in \Strat(v)$ with $\Ss' \absorbneq \Ss$.
It is further \emph{strictly} absorption-dominant, if there is no other strategy $\Ss' \in \Strat(v)$ with $\Ss' \absorb \Ss$, so no other strategy is absorption-equivalent to $\Ss$.
\end{definition}

Notice that absorption is simply the inverse pointwise order on the edge profiles.
In particular, $\Ss_1 \equiv \Ss_2$ if, and only if, $\ep {\Ss_1} = \ep {\Ss_2}$.
We next aim at understanding the relation between
(strictly) absorption-dominant strategies and the standard notion of positional strategies.
As a starter, we show that absorption-dominant strategies are not necessarily positional
(cf.\ \cite{GraedelTan20} for a similar example).

\begin{Example}\label{ex:Weakpos}
Consider the strategy $\Ss$ as depicted in \Cref{fig:RunningStrategy}.
It is not positional, as the choice for position $v'$ is not unique (both $e$ and $f$ occur in $\Ss$).
It is, however, absorption-dominant.
As there are two paths to $v'$, every strategy must either use $e$ or $f$ twice, or use both edges.
If $e$ (or $f$) is used twice, then the strategy cannot absorb $\Ss$, and one can verify that $\Ss$ absorbs all strategies using both $e$ and $f$.

It is not strictly absorption-dominant, as we obtain an absorption-equivalent strategy by switching the two branches in the depiction of $\Ss$, so that $e$ is used after $c$, and $f$ after $b$.
\end{Example}

Strategies such as the one in \Cref{fig:RunningStrategy} are not positional, but satisfy the weaker property that within each \emph{play}, the strategy makes a unique decision for each position $v \in V_0$.
This notion of strategies has been introduced as \emph{persistent} strategies in \cite{MarcinkowskiTruderung02} in the context of LTL on game graphs and has been further studied in \cite{Duparc03}.
Persistent strategies have also been called \emph{weakly positional} in \cite{GraedelTan20}.

We say that a strategy \emph{plays positionally} from a position $v \in V_0$ if the strategy makes a unique choice at position $v$.
A strategy that plays positionally from all positions in $V_0$ is positional.
With this notation, we now clarify the relation between the different notions of strategies; a summary is shown in \Cref{figStrategyClasses}.
We first observe that if a strategy $\Ss$ does not play positionally from $v$, we can always obtain a strategy $\Ss'$ with $\Ss' \absorb \Ss$ by swapping the choices at $v$. This leads to the first result (see \cite{GraedelLucNaa21} for details):

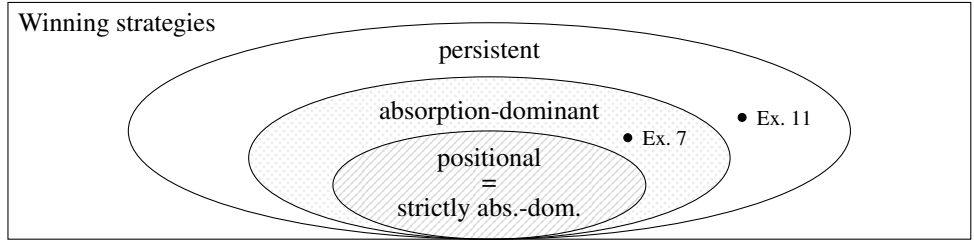
\begin{figure}[t]
\centering
\begin{tikzpicture}[font=\small, every label/.style={font=\scriptsize}, xscale=0.8,yscale=0.9]
\draw [pattern=crosshatch dots, pattern color=lightgray!30]  (0,.4) ellipse (4cm and 1.2cm);
\node[anchor=north] at (0,1.4) {absorption-dominant};
\draw [pattern=north east lines, pattern color=lightgray!60] (0,0) ellipse (2.6cm and .8cm) node [align=center,yshift=0pt] {positional \\[-.3em] = \\[-.2em] strictly abs.-dom.} ;
\draw (0,.8) ellipse (6cm and 1.6cm);
\node[anchor=north] at (0,2.3) {persistent};
\draw (-8,-.8) rectangle (8,2.7);
\node [anchor=north west] at (-8,2.7) {Winning strategies};
\node [dot,label={right:Ex.~\ref{ex:Weakpos}}] at (2.3,0.7) {};
\node [dot,label={right:Ex.~\ref{ex:WeakposNotDominant}}] at (4.2,1.0) {};
\end{tikzpicture}
\vspace{-.3em}
\caption{Venn diagram depicting the relation of various classes of winning strategies.}
\label{figStrategyClasses}
\end{figure}

\begin{proposition}
Strictly absorption-dominant strategies coincide with positional strategies.
\end{proposition}

Towards persistent strategies, we observe, by applying a simple combinatorial fact known as Dickson's Lemma to edge profiles, that there are only finitely many absorption-dominant strategies from $v$ up to absorption-equivalence.
This leads to the following result:

\begin{proposition}\label{stratInfinitePositional}
Let $\Ss \in \WinStrat_\Gg(v)$ be absorption-dominant from $v$, and let $w \in V_0$ be a position.
If $w$ occurs infinitely often in $\Ss$, then $\Ss$ plays positionally from $w$.
\end{proposition}

\begin{proof}[Proof sketch]
Consider the infinitely many substrategies at occurrences of $w$ in $\Ss$.
By Dickson's Lemma, there is one such substrategy $\Ss_w$ such that infinitely many of the substrategies are absorption-equivalent to $\Ss_w$.
In particular, $\Ss_w$ only uses edges that occur infinitely often in $\Ss$.
By positional determinacy, there is thus a positional winning strategy $\Ss_{\text{pos}}$ from $w$ with the same property.
If $\Ss$ would not be positional, then modifying $\Ss$ to always play $\Ss_{\text{pos}}$ from $w$ would result in a strategy $\Ss' \absorbneq \Ss$, a contradiction.
\end{proof}

With this important insight, we can deduce that the absorption-dominant winning strategies (from some position $v$) are a (strict) subset of the persistent strategies:
An absorption-dominant strategy must play positionally from positions that occur infinitely often by \Cref{stratInfinitePositional}; positions that occur finitely often cannot be visited twice in one play (this would be a redundant repetition).

\begin{corollary}\label{stratDominantWeakpos}
Every absorption-dominant winning strategy in $\Gg$ is persistent.
\end{corollary}

\begin{Example}\label{ex:WeakposNotDominant}
For strictness, consider the following game (a modified part of \Cref{fig:RunningGame}):

\begin{center}
\begin{tikzpicture}[scale=0.9,yscale=0.7,xscale=0.6,baseline={(0,0)}]
\node [p1] at (0,0) (1) {$v$};
\node [p1] at (1.5,1) (2a) {};
\node [p0,F] at (3,0) (3) {};
\node [p0,F] at (4.5,1) (4a) {};
\path [arr] (1) edge node {} (2a) (1) edge node [below left] {} (3)
      (2a) edge node {} (3)
      (3) edge node {$a$} (4a)
      (3) edge [in=-60,out=-20,min distance=1cm] node {$b$} (3)
      (4a) edge [loop right] node {} ()
      ;
\end{tikzpicture}
\hspace{.5cm}
\parbox[c]{3.5cm}{
$\ecount {\Ss_1} {a,b} = (2,0)$, \\
$\ecount {\Ss_2} {a,b} = (0,\infty)$, \\
$\ecount {\Ss_3} {a,b} = (1,\infty)$.
}
\end{center}

\noindent
Due to the self-loop $b$, only the positional strategies $\Ss_1$ (always take $a$) and $\Ss_2$ (always take $b$) are absorption-dominant from $v$.
The strategy $\Ss_3$ that, depending on Player~1's choice, either takes edge $a$ or loops indefinitely using edge $b$ is persistent, but not absorption-dominant:
it is strictly absorbed by $\Ss_2$.
\end{Example}

As a consequence of \Cref{stratDominantWeakpos}, all moves after the first repeated position are determined by persistence and we can thus represent absorption-dominant strategies in a compact way.

\begin{corollary}
Let $\Gg$ be a game with $n = |V|$ positions.
Every winning strategy $\Ss \in \WinStrat_\Gg(v)$ that is absorption-dominant from $v$
can be uniquely represented by a subtree of the tree unraveling of height at most $n$.
In particular, the number of absorption-dominant winning strategies is finite.
\end{corollary}

\section{A Whirlwind Tour of Semiring Semantics}

This section gives an overview on semiring semantics for fixed-point logics,
with a focus on the semirings relevant for the case study.
For a complete account, we refer to \cite{DannertGraNaaTan21}.

\subsection{Semirings}\label{sect:semirings}

Semirings are algebraic structures with two binary operations, usually denoted $+$ and $\bcdot$,
which we use to interpret the logical connectives $\lor$ and $\land$.
While semirings are very general structures, we make additional assumptions to ensure well-defined and meaningful semiring semantics for logics with fixed-point operators, following the definitions in \cite{DannertGraNaaTan21}.

\begin{definition}
A \emph{commutative semiring} is an algebraic structure 
$(\K,+,\bcdot,0,1)$, with $0\neq1$,  such that $(\K,+,0)$
and $(\K,\bcdot,1)$ are commutative monoids, $\bcdot$
distributes over $+$, and $0\bcdot a=a\bcdot 0=0$.
It is \emph{idempotent} if $a+a = a$ for all $a \in \K$.
\end{definition}

All semirings we consider are commutative, so we omit \emph{commutative} in the following.
Towards fixed-point logic, we compute least and greatest fixed points with respect to the natural order $\le_\K$ (see below) and to ensure that they exist, we require $\le_\K$ to be a complete lattice (in fact, suprema and infima of chains would suffice, but in idempotent semirings this is equivalent).
We additionally impose a natural \emph{continuity} requirement which is crucial to our proofs, but does not seem to be a strong restriction in practice (we are not aware of any natural complete-lattice semirings that are not continuous).
Regarding notation, a \emph{chain} is a totally ordered set $C \subseteq \K$ and we write
$a \circ C = \{ a \circ c \mid c \in C\}$ for $a \in \K$.

\begin{definition}
In an idempotent semiring $(\K,+,\bcdot,0,1)$, the \emph{natural order} $\le_\K$ is the partial order defined by $a \le_\K b \Leftrightarrow a+b=b$.
We say that $\K$ is \emph{fully continuous} if $\le_\K$ is a complete lattice (with supremum $\Sup$ and infimum $\Inf$) and for all non-empty chains $C \subseteq \K$, elements $a \in \K$ and $\circ \in \{+,\bcdot\}$,
\[
    \Sup (a \circ C) = a \circ \Sup C,
    \quad \text{and} \quad
    \Inf (a \circ C) = a \circ \Inf C.
\]
A semiring homomorphism $h \colon \K_1 \to \K_2$ on fully-continuous semirings is \emph{fully continuous} if $h(\Sup C) = \Sup h(C)$ and $h(\Inf C) = \Inf h(C)$ for all non-empty chains $C \subseteq \K_1$.
\end{definition}

By the Knaster-Tarski theorem, every $\le_\K$-monotone function $f \colon \K \to \K$ on a fully-continuous semiring has a least fixed point $\lfp(f)$ and a greatest fixed point $\gfp(f)$ in $\K$,
and this suffices to guarantee well-defined semantics of fixed-point logics.
However, from a provenance perspective we further want this semantics to be meaningful in the sense that the value of a formula provides insights into why the formula holds.
It turns out that this is the case if we additionally require the semiring to be absorptive \cite{DannertGraNaaTan21}.

\begin{definition}
A semiring $\K$ is \emph{absorptive} if $a+ab = a$ for all $a,b \in \K$.
\end{definition}

We remark that absorption is equivalent to $\K$ being \emph{0-closed} or \emph{bounded} \cite{Mohri02},
that is, $1+a = 1$.
If $\K$ is idempotent, then absorption is further equivalent to multiplication being decreasing,
that is, $a \bcdot b \le_\K a,b$.
Clearly, every absorptive semiring is idempotent and thus partially ordered by $\le_\K$, with $1$ as top element.
If we additionally assume full continuity, we can extend any absorptive semiring by an infinitary power operation
$a^\infty = \Inf_{n \in \N} a^n$ with natural properties such as $a \bcdot a^\infty = a^\infty$ and $(a+b)^\infty = a^\infty + b^\infty$.

\begin{Example}
Examples of semirings used in provenance analysis of databases and logics \cite{GreenKarTan07,GreenTan17,GraedelTan20} are:
\begin{itemize}
\item The \emph{Boolean semiring} $\Bool=(\{\Bzero,\Bone\},\vee,\wedge,\Bzero,\Bone)$ is the standard habitat of 
logical truth. It is absorptive and (trivially) fully continuous.
\item $\Nat=(\Nat,+,\cdot,0,1)$ is used for counting evaluation strategies for a logical statement.
It is not absorptive and hence not well suited for fixed-point logics.
\item The \emph{Viterbi} semiring $\Vit=([0,1],\max,\cdot,0,1)$
is used to compute \emph{confidence scores}
for logical statements. It is isomorphic to the \emph{tropical} semiring $\Trop=(\mathbb{R}_{+}^{\infty},\min,+,\infty,0)$ 
which is used for measuring the cost of 
evaluation strategies.
Both are absorptive and fully continuous.

\item The \emph{min-max} semiring $(A, \max, \min, a, b)$ on a totally ordered set $(A,\leq)$
with least and greatest elements $a$ and $b$ can be used to model access privileges. It is absorptive and fully continuous. \cornerhere
\end{itemize}
\end{Example}

From now on, all semirings we consider are commutative, absorptive and fully continuous.
Besides the application semirings listed above, we are particularly interested in universal semirings of polynomials to represent abstract information.
We can then use fully-continuous homomorphisms to specialize the computed information to application semirings as needed, as these homomorphisms preserve fixed points.

The common examples of semirings of polynomials $\N[X]$ and formal power series $\N^\infty \ps X$,
as used for provenance analysis of FO and Datalog in \cite{GreenKarTan07, GraedelTan17}, are not absorptive and hence not well-suited for fixed-point logic.
Instead, we rely on semirings of (generalized\footnote{The definition we use here generalizes the notion of absorptive polynomials in \cite{DeutchMilRoyTan14} by allowing $\infty$ as exponent.}) \emph{absorptive polynomials}.
We summarize the definition and main properties given in \cite{DannertGraNaaTan21}.
Essentially, an absorptive polynomial such as $ab^3 + c^\infty$ is a sum of monomials over a finite set of variables $X$,
but without coefficients and with exponents from $\N \cup \{\infty\}$ (where $n < \infty$ for $n \in \N$).
Monomial multiplication is defined as usual (with $n+\infty = \infty$).

The key ingredient is absorption among monomials.
We say that a monomial $m_1$ \emph{absorbs} $m_2$, if all its exponents are smaller (or equal).
Formally, $m_1 \absorb m_2$ if $m_1(x) \le m_2(x)$ for all $x \in X$, where $m_1(x)$ denotes the exponent of $x$ in $m_1$.
For example, $ab^2 \absorb a^\infty b^2$ and $a \absorb ab$, but $a^2b$ and $ab^2$ are incomparable.
In an absorptive polynomial, we omit all monomials that would be absorbed, so absorptive polynomials
are $\absorb$-\emph{antichains} of monomials. Consequently, addition and multiplication
are defined as usual, but afterwards we only keep the $\absorb$-maximal monomials.
For example, $(ab^2 + a^2b) \bcdot a^\infty = a^\infty b^2 + a^\infty b = a^\infty b$.

We write $\Sinf[X]$ for the semiring of absorptive polynomials over the finite variable set $X$.
The $0$ and $1$-elements are the empty polynomial and the single monomial $1$ (with all zero exponents).
This defines an absorptive, fully-continuous semiring \cite{DannertGraNaaTan21}.
In fact, $\Sinf[X]$ is the most general such semiring:

\begin{theorem}[Universal Property, \cite{DannertGraNaaTan21}]
\label{universality}
Every mapping $h \colon X \to \K$ into an absorptive, fully-continuous semiring $\K$
uniquely extends to a fully-continuous semiring homomorphism $h \colon \Sinf[X] \to \K$.
\end{theorem}

\subsection{Logic}

We consider here the fixed-point logic LFP that extends first-order logic FO by least and greatest fixed-point formulae of the form $\psi(\ty) = [\lfp R \tx.\ \phi(R,\tx)](\ty)$ and $\psi(\ty) = [\gfp R \tx.\ \phi(R,\tx)](\ty)$. Here, $R$ is a relation symbol occurring only positively in $\phi$ and $\tx,\ty$ are variable tuples of matching arity.
Given a (Boolean) model $\AA$ and a tuple $\ta$ of elements of $\AA$,
the formula $\psi(\ta)$ holds in $\AA$, denoted $\AA \models \psi(\ta)$,
if $\ta$ is contained in the least (or greatest) fixed point of the operator
$F_\phi \colon R \mapsto \{ \ta \mid \AA \models \phi(R,\ta) \}$
that maps a relation $R$ to the relation consisting of those tuples for which $\phi$ holds.
For more background and a precise definition, we refer to \cite{Graedel+07}.

In order to generalize Boolean semantics to semiring semantics,
we first adapt the notion of a model $\AA$.
Instead of determining for each literal whether it is true or false in $\AA$,
we assign to each literal a semiring value, interpreting $0$ as \emph{false} and all other values as \emph{nuances of true}.
Special care is required to ensure that the assignment is consistent with respect to opposing literals (this is not always necessary, but often desirable).
In the following, let $\K$ be a semiring, $A$ a finite universe and $\tau$ a relational signature
(we drop $A$ and $\tau$ if clear from the context).
We denote the set of (instantiated) literals as
\begin{align*}
    \Lit_{A,\tau} ={} &\{ R \ta, \neg R \ta \mid R \in \tau \text{ of arity $k$}, \ta \in A^k \}
    \,\cup\, \{\ta = \tb, \ta \neq \tb \mid \ta,\tb \in A^k\}.
\end{align*}
Given a literal $L$, we write $\neg L$ for the opposing literal (identifying $\neg \neg L$ and $L$).
The role of the Boolean model $\AA$ is then replaced by a semiring interpretation $\pi$ that assigns semiring values to all literals.

\begin{definition}\label{defKInterpretation}
Let $\K$ be a semiring. A \emph{$\K$-interpretation} (over finite $A$ and $\tau$) is an assignment $\pi\colon \Lit_{A,\tau} \to \K$ that maps true (in)equalities to $1$ and false (in)equalities to $0$.
We say that $\K$ is \emph{model-defining}, if for each literal $L$, exactly one of $\pi(L)$ and $\pi(\neg L)$ is $0$.
\end{definition}

We lift $\pi$ to LFP-formulae in negation normal form, resulting in a semiring value $\pi \ext \psi$, by interpreting $\land$ and $\lor$ as semiring operations $\bcdot$ and $+$, and similarly quantifiers as products or sums over the (finite) universe.
For fixed-point formulae, we consider the induced operator $F_\phi$ analogous to the Boolean case (but acting on mappings $A^k \to \K$ instead of relations $R \subseteq A^k$) and compute its fixed point in the semiring $\K$.
We refer to \cite{DannertGraNaaTan21} for a formal definition;
the details are not relevant for this paper, as we will only consider a fixed formula $\psi$ (see \Cref{sec:TheFormula}).
An important property of the resulting semantics is that they are preserved by fully-continuous semiring homomorphisms,
in particular by polynomial evaluation in $\Sinf[X]$ due to \Cref{universality} (but not by polynomial evaluation of $\Nat[X]$ or formal power series!).

\section{Case Study: Computing Strategies with Semiring Semantics}
\label{sec:TheFormula}

This section connects the previous sections on semiring semantics and absorption-dominant strategies.
We focus on the formula for the winning region in a Büchi game and show that its value under semiring semantics can be understood in terms of (absorption-dominant) winning strategies.

\subsection{The Semiring Interpretation}

We want to use semiring semantics to analyze moves in winning strategies.
For this reason, we label edges with indeterminates $X$ (cf.\ \Cref{fig:RunningGame}) and use an $\Sinf[X]$-interpretation $\pitrack$ to track moves (i.e., edge literals $Euv$) via their indeterminates.
We assume the game graph to be fixed and do not wish to track information about the winning set $F$ or the active player at a certain node, hence we simply map all other literals over $\tau = \{E,F,V_0,V_1\}$ (e.g., $F v$, $V_0 v$, $\neg Euv$) to $0$ or $1$, depending on whether they are true or false in the game.
The resulting interpretation is almost Boolean and hence behaves very similar to the original game, except that we remember which edges are used in the evaluation of a formula.

\begin{definition}
Let $G = (V, V_0, V_1, E, F)$ be a Büchi game and
let $X = \{ X_{vw} \mid vw \in E \}$ be a set of indeterminates for all edges.
We define the $\Sinf[X]$-interpretation $\pitrack$ as follows (depending on $G$):
\[
    \pitrack(Evw) = X_{vw} \text{ for all } vw \in E, \quad
    \pitrack(L) = \begin{cases}
    1, &\text{ if } G \models L, \\
    0, &\text{ if } G \not\models L,
    \end{cases} \; \text{ for all other literals $L \in \Lit_{V,\tau}$.}
\]
\end{definition}

\subsection{The Formula}

It is well known that the winning region (of Player~0) in a Büchi game is definable in fixed-point logic.
Intuitively, the winning region is the largest set $Y$ such that from each position in $Y$,
Player~0 can enforce a visit to $Y \cap F$ (after at least one move).
In LFP, we can express the winning region by the following formula (see, e.g., \cite{CanavoiGraLesPak15,Walukiewicz02}):
\begin{align*}
  \FormulaWin(x) \coloneqq{} &\big[\gfp Y y.\ [\lfp Z z.\ \phi(Y,Z,z)](y) \big](x), \\
  \phi(Y,Z,z) \coloneqq{} &\Big(Fz \;\land\; ((V_0 z \land \E u (Ezu \land Yu)) \lor (V_1 z \land \A u (Ezu \to Yu)))\Big) \\
  {}\lor{} &\Big(\neg Fz \;\land\; ((V_0 z \land \E u (Ezu \land Zu))
  \lor (V_1 z \land \A u (Ezu \to Zu)))\Big).
\end{align*}

Given a $\K$-interpretation $\pi$ for a Büchi game $\Gg = (V,V_0,V_1,E,F)$, semiring semantics of the above formula induce%
\footnote{Here we first translate $Ezu \to Yu$ to the formula $\neg Ezu \lor (Ezu \land Yu)$ in negation normal form.}
the following fixed-point computation.
To simplify the presentation, we introduce two families of variables,
$\mathbf Y = (Y_v)_{v \in V}$ and $\mathbf Z = (Z_v)_{v \in V}$ that take values in $\K$.
We can then express the semiring semantics as $\pi \ext {\FormulaWin(v)} = Y^*_v$ where
$\mathbf Y^* = (Y_v^*)_{v \in V}$ is the \emph{greatest} solution to the equation system
\[
    \mathbf Y = \mathbf Z^*(\mathbf Y)
\]
where, in turn, $\mathbf Z^*(\mathbf Y)$ is the \emph{least} solution, given values $\mathbf Y = (Y_v)_{v \in V}$, to the equation system consisting of the following equation for all $v \in V$:
\begin{align*}
    Z_v ={} &\pi(Fv) \bcdot \Big((\pi(V_0 v) \bcdot \sum_{w \in V} (\pi(Evw) \bcdot Y_w)) + (\pi(V_1 v) \bcdot \prod_{w \in V} (\pi(\neg Evw) + \pi(Evw) \bcdot Y_w))\Big) \\
    {}+{} &\pi(\neg Fv) \bcdot \Big((\pi(V_0 v) \bcdot \sum_{w \in V} (\pi(Evw) \bcdot Z_w)) + (\pi(V_1 v) \bcdot \prod_{w \in V} (\pi(\neg Evw) + \pi(Evw) \bcdot Z_w))\Big)
\end{align*}
For most of this paper, we use the interpretation $\pitrack$ to track only \emph{moves} of winning strategies.
As $\pitrack$ maps most of the literals to $0$ or $1$, we can simplify the equations depending on $v$:

\begin{center}
\renewcommand{\arraystretch}{1.4}
\begin{tabular}{c|c|c}
& $v \in F$ & $v \notin F$ \\ \hline
$v \in V_0$ &
    $\displaystyle Z_v = \sum_{w \in vE} \pi(Evw) \bcdot Y_w$ &
    $\displaystyle Z_v = \sum_{w \in vE} \pi(Evw) \bcdot Z_w$ \\
$v \in V_1$ &
    $\displaystyle Z_v = \prod_{w \in vE} \pi(Evw) \bcdot Y_w$ &
    $\displaystyle Z_v = \prod_{w \in vE} \pi(Evw) \bcdot Z_w$
\end{tabular}
\end{center}

A good way to think about (and compute) the least and greatest solutions is the fixed-point iteration.
The idea is to start with each $Z_v$ set to the least element of the semiring, then apply the above equations (i.e., the induced operator $F_\phi$) to compute a next, larger semiring value and repeat this process until a fixed-point is reached (notice that the iteration can also be infinite, we then continue with the supremum/infimum).

\begin{Example}\label{exComputation}
Recall the simple game from \Cref{ex:redundantMove}
$\big($\!\!
\begin{tikzpicture}[baseline,font=\scriptsize]
\node [p0,label={below:$v$},anchor=base,yshift=.1cm] (0) {};
\node [p0,F,label={below:$w$},right of=0] (1) {};
\draw [arr]
    (0) edge [loop left] node {$a$} (0)
    (1) edge [loop right] node {$c$} (1)
    (0) edge node {$b$} (1);
\end{tikzpicture}
$\!\!\big)$.

\newcommand{\vv}[2]{\begin{pmatrix}#1 \\ #2\end{pmatrix}}
\newcommand{\tY}{\mathbf Y}
\newcommand{\tZ}{\mathbf Z}
\newcommand{\arrF}{\xmapsto{\!\!\!F_{\pitrack}^\phi\!}}

Using the interpretation $\pitrack$ corresponding to the edge labels,
we obtain the following fixed-point iteration.
We write the tuples $\tY$ and $\tZ$ as vectors $({Y_v} \; {Y_w})^T$ and $({Z_v} \; {Z_w})^T$.
\begin{center}
\begin{tikzpicture}[font=\small,node distance=1.5cm]
\node (Y) {$\tY:$};
\node (Z) [below=.7cm of Y] {$\tZ:$};

\node [right=0cm of Y] (Y1) {$\vv 1 1$};
\node [xshift=.7cm] at (Z -| Y1) (Z11) {$\vv 0 0$};
\node [right of=Z11] (Z12) {$\vv 0 c$};
\node [right of=Z12] (Z13) {$\vv {bc} c$};

\node [xshift=.7cm] (Y2) at (Y1 -| Z13) {$\vv {bc} c$};
\node [xshift=.7cm] at (Z -| Y2) (Z21) {$\vv 0 0$};
\node [right of=Z21] (Z22) {$\vv 0 {c^2}$};
\node [right of=Z22] (Z23) {$\vv {bc^2} {c^2}$};

\node [xshift=.7cm] (Y3) at (Y1 -| Z23) {$\vv {bc^2} {c^2}$};
\node [xshift=.7cm] at (Z -| Y3) (Z31) {$\dots\vphantom{\vv x x}$};

\node [right=2cm of Y3] (Ystar) {$\tY^* = \vv {bc^\infty} {c^\infty}$};

\path [draw,|->,every node/.style={anchor=base,yshift=4pt}]
    (Z11) edge node {$F_\phi$} (Z12)
    (Z12) edge node {$F_\phi$} (Z13)
    (Z21) edge node {$F_\phi$} (Z22)
    (Z22) edge node {$F_\phi$} (Z23)
;

\path [draw,->,short=-3pt]
    (Y1) edge (Z11)
    (Z13) edge (Y2)
    (Y2) edge (Z21)
    (Z23) edge (Y3)
    (Y3) edge (Z31)
;

\path[draw,<-,decorate,decoration={snake,amplitude=1.5pt,segment length=1.5mm,pre=lineto,pre length=3pt,post=lineto,post length=0pt}]
    (Ystar) -- ($(Y3)+(1.7cm,0)$);
\end{tikzpicture}
\end{center}

We obtain the overall result $\pitrack \ext {\FormulaWin(v)} = Y_v^* = \Inf_n bc^n = bc^\infty$ corresponding to the unique absorption-dominant strategy using edge $b$ once and $c$ infinitely often (cf.\ \Cref{ex:redundantMove}).
\end{Example}

\subsection{Connection to Strategies}

By mapping edges to semiring values, we can track edges through the fixed-point computation.
In \Cref{exComputation}, the resulting semiring value revealed how often each edge is used in the unique absorption-dominant winning strategy.
We now generalize this observation.
For simplicity, we only consider $\K$-interpretations $\pi$ that are \emph{edge tracking} for a given game $\Gg$.
That is, they may assign arbitrary values to positive edge literals $Evw$, but all other literals are mapped to $0$ or $1$ in accordance with $\Gg$.
To make the connection to strategies explicit, we first define semiring values for strategies based on their edges.

\begin{definition}
Let $\Ss$ be a strategy in a Büchi game $\Gg = (V, V_0, V_1, E, F)$.
Let $\K$ be an absorptive, fully-continuous semiring and $\pi$ an edge-tracking $\K$-interpretation on $\Gg$.
The \emph{$\K$-value} of $\Ss$ is the product of the values for all edges appearing in $\Ss$.
Formally,
\[
    \pi \ext \Ss \coloneqq \prod_{vw \in E} \pi(Evw)^{\ecount \Ss {vw}}.
\]
\end{definition}

The semiring value of $\FormulaWin$ can then be expressed as the sum over the values of all winning strategies.
The proof of this result (see \cite{GraedelLucNaa21}) is based on a similar result for LFP model-checking games in \cite{DannertGraNaaTan21} which is rather involved due to both infinite strategy trees and infinite fixed-point iterations.

\begin{restatable}[Sum of Strategies]{theorem}{SumOfStrategies}
\label{thmSumOfStrategiesTracking}
Let $\Gg$ be a Büchi game and $v$ a position in $\Gg$.
Let $\K$ be an absorptive, fully-continuous semiring and $\pi$ an edge-tracking $\K$-interpretation.
Then,
\[
    \pi \ext {\FormulaWin(v)} = \sum \big\{ \pi \ext \Ss \;\big|\;
    \text{$\Ss \in \WinStrat_\Gg(v)$ is absorption-dominant from $v$} \big\}.
\]
\end{restatable}

It is in fact this central result that motivated the notion of \emph{absorption-dominant} strategies.
However, as we have already discussed, these may also be interesting in their own right if one is interested in minimal winning strategies.

\begin{Example}
For the edge-tracking interpretation $\pitrack$ induced by the edge labels in  \Cref{fig:RunningGame}, we obtain
\begin{align*}
    \pitrack \ext {\FormulaWin(v)} ={}
    &(abcd)^\infty + 
    abc \, e^2 h^2 (gkm)^\infty +
    abc \, f^2 (gkm)^\infty +
    abc \, ef h (gkm)^\infty.
\end{align*}

There are four monomials, corresponding to four equivalence classes of absorption-dominant strategies.
Each monomial reveals the edges that appear in the corresponding strategies, so we see that the first three monomials belong to positional (and hence uniquely defined) strategies.
The last monomial belongs to the non-positional strategy shown in \Cref{fig:RunningStrategy} (and its switched version, see \Cref{ex:Weakpos}).
The values of all other strategies are strictly absorbed by one of these monomials.
\end{Example}

\subsection{Strategy Analysis}
\label{sec:StrategyTracking}

For the general case, we fix a Büchi game $\Gg$ and focus on the $\Sinf[X]$-interpretation $\pitrack$ with $X = \{ X_{uv} \mid u,v \in \Gg \}$.
The values $\pitrack \ext \Ss$ are monomials and we can read off the number of occurrences of each edge in $\Ss$ from the exponents, i.e., the monomial is a representation of the edge profile $\ep \Ss$.
In particular, $\pitrack \ext {\Ss_1} \absorb \pitrack \ext {\Ss_2}$ if,
and only if, $\Ss_1 \absorb \Ss_2$.
The fact that absorptive polynomials are always finite \cite{DannertGraNaaTan21} is thus another way to see that the number of absorption-dominant strategies is finite.

What can we learn from the polynomial $\pitrack \ext {\FormulaWin(v)}$?
First, $\pitrack \ext {\FormulaWin(v)} \neq 0$ holds if, and only if, Player~0 has a winning strategy from $v$.
By \Cref{thmSumOfStrategiesTracking}, we can further derive information about all absorption-dominant strategies.
More precisely, we learn which edges each absorption-dominant strategy uses and how often they appear in the strategy tree.
Knowing the edge profile immediately reveals whether the strategy is positional and what the positional choices are.
By counting monomials, we can thus count the positional strategies, as well as the absorption-dominant strategies up to absorption-equivalence.

We can further answer questions such as: can Player~0 still win if we remove edge $e$?
This is the case if, and only if, the polynomial $\pitrack \ext {\FormulaWin(v)}$ contains a monomial without the variable $X_e$ (if there is a winning strategy without $e$, then there is also an absorption-dominant strategy and hence a monomial without $X_e$).
Going further, a more interesting question is: can Player~0 still win if edge $e$ may only be used finitely often in each play?
The answer is not immediately obvious.
Consider for example the strategy $\Ss$ in \Cref{fig:RunningStrategy}.
The edge $k$ occurs infinitely often in the strategy tree and we get $\pitrack \ext \Ss = abcefh g^\infty k^\infty m^\infty$.
However, $k$ is clearly played only once in each play consistent with $\Ss$, whereas edge $m$ is played infinitely often.
We cannot distinguish edges $k$ and $m$ just from $\pitrack \ext \Ss$, but we can do so if we compute $\pitrack \ext {\FormulaWin(w)}$ for all positions $w \in V$, by the following criterion.

\begin{proposition}
Let $\Ss \in \WinStrat_\Gg(v)$ be absorption-domi\-nant from $v$, and let $e = uw \in E$ be an edge with $\ecount \Ss e = \infty$.
Then there is a unique positional strategy $\Ss_w \in \WinStrat_\Gg(w)$ such that $\pitrack \ext {\Ss_w} \absorb \pitrack \ext \Ss$.
Moreover, $\Ss$ admits a play in which $e$ occurs infinitely often if, and only if, $e$ occurs
in $\Ss_w$.
\end{proposition}

\begin{Example}
Consider the strategy $\Ss$ in \Cref{fig:RunningStrategy} with $\pitrack \ext \Ss = abcefh g^\infty k^\infty m^\infty$ and the edge $k$ from $u$ to $w$.
Since edge $n$ does not occur in $\pitrack \ext \Ss$, the only winning strategy $\Ss_w$ from $w$ we need to consider is the strategy that always stays at $w$, with $\pitrack \ext {\Ss_w} = m^\infty \absorb \pitrack \ext \Ss$.
As $k$ does not occur in $\Ss_w$, we conclude that it occurs only finitely often (and hence at most once) in each play consistent with $\Ss$.

If, on the other hand, we consider edge $m$ (which also leads to position $w$), we see that $m$ occurs in $S_w$ and we can thus conclude that $\Ss$ contains a play visiting $m$ infinitely often.
\end{Example}

Summarizing the results of this section, we see that semiring semantics in $\Sinf[X]$ is very informative and allows us to derive important information about the winning strategies.

\begin{corollary}\label{corPolynomialInformation}
From the polynomial $\pitrack \ext {\FormulaWin(v)}$, we can efficiently (in the size of the polynomial) derive the following information:\\
-\quad whether Player $0$ wins from $v$,\\
-\quad the edge profiles of all absorption-dominant winning strategies from $v$,\\
-\quad the number and precise shape of all positional winning strategies from $v$,\\
-\quad whether Player $0$ can still win from $v$ if only a subset of the edges is allowed.

\noindent From the polynomials $\pitrack \ext {\FormulaWin(v)}$, for all positions $v$,
we can further derive for each absorption-dominant strategy and each edge, how often the edge can occur in a play consistent with the strategy.
\end{corollary}

Regarding the complexity, we remark that it is not obvious how to compute $\pitrack \ext {\FormulaWin(v)}$, as depending on the semiring, the fixed-point iteration may be infinite (even if the game is finite).
However, it has been shown in \cite{Naaf21} that in absorptive, fully-continuous semirings, fixed points of polynomial systems, such as the ones induced by $\FormulaWin(v)$, can be computed using a number of semiring operations (including the infinitary power operation) that is polynomial in the size of the game.

Notice, however, that for $\pitrack \ext {\FormulaWin(v)}$, the number of monomials and hence the cost to compute semiring operations can become exponentially large.
This is not avoidable, as the number of (positional) winning strategies can be exponential.
More efficient algorithms can be obtained by modifying $\pitrack$ to track only some of the edges (mapping the rest to $1$).
The resulting interpretation remains edge-tracking, so the Sum-of-Strategies Theorem still applies and we can see which sets of tracked edges (if any) are required for a winning strategy.
Even then, some of the questions in \Cref{corPolynomialInformation} can also be solved by direct methods which may be more efficient.
Thus, the main benefit of semiring semantics in $\Sinf[X]$ does not lie in a more efficient method to compute 
some specific winning strategy, but rather in providing a general and compact description of all important strategies at once, from which we can directly derive the answers to many different questions concerning the strategy analysis of a Büchi game.

\subsection{Reverse Analysis}

Instead of tracking strategies in a fixed game, we may also ask questions such as:
assuming Player~1 wins from $v$, what are minimal modification to $\Gg$ such that instead Player~0 wins?
The generality of semiring semantics enables us to answer such questions by choosing appropriate semirings and interpretations.

More precisely, let $\Gg = (V,V_0,V_1,E,F)$ be a Büchi game and $v \in \Gg$ a position from which Player~1 wins.
Let $E^- \subseteq E$ and $E^+ \subseteq V^2 \setminus E$ be sets of edges we are allowed to delete or add, respectively.
We call a set of edges in $E^\pm \coloneqq E^- \cup E^+$ a \emph{repair} if Player~0 wins when these edges are deleted or added.
Our goal is to determine all (preferably minimal) repairs.
We achieve this by evaluating $\FormulaWin(v)$ in a modified polynomial semiring, similar to the computation of repairs for database queries in \cite{XuZhaAlaTan18},
except that here we need absorptive polynomials to deal with fixed points.

\bigskip\noindent\textbf{Dual-Indeterminates.}
To track negative information, such as the absence of an edge, we follow the approach in \cite{GraedelTan17,XuZhaAlaTan18,DannertGraNaaTan21} and extend our semiring by dual-indeterminates $\nnX = \{ \nnx \mid x \in X \}$.
The idea is to label a literal and its negation by corresponding indeterminates $x$ and $\nnx$.
We must then avoid monomials such as $x\nnx$, as they represent contradictory information.
To this end, we consider the quotient of $\Sinf[X \cup \nnX]$ with respect to the congruence generated by $x \cdot \nnx = 0$ for $x \in X$
and refer to the resulting quotient semiring as dual-indeterminate absorptive polynomials $\Sinf[X,\nnX]$.
This semiring inherits most of the properties of $\Sinf[X]$.
Most importantly, any assignment $h \co X \cup \nnX$ that respects dual-indeterminates, i.e., $h(x) \cdot h(\nnx) = 0$, lifts to a fully-continuous homomorphism analogous to \Cref{universality}.

We then replace $\pitrack$ with an $\Sinf[X,\nnX]$-interpretation $\pirev^\pm$ for $X = \{ X_e \mid e \in E^\pm\}$:
if $vw \in E^\pm$, we set $\pirev^\pm(Evw) = X_{vw}$ and $\pirev^\pm(\neg Evw) = \nn{X_{vw}}$,
all other literals are mapped to $0$ or $1$ according to $\Gg$.
Notice that $\pirev^\pm$ is not model-defining, but still satisfies $\pirev^\pm(L) \bcdot \pirev^\pm(\neg L) = 0$ for all literals $L$.

\bigskip\noindent\textbf{Back and Forth between Monomials and Models.}
Let $X^\pm = \{ X_e \mid e \in E^+ \} \cup \{ \nn{X_e} \mid e \in E^- \}$.
Given $Y \subseteq X^\pm$, we further write $E(Y) = \{ e \mid X_e \in Y \text{ or } \nn{X_e} \in Y \}$ for the set of edges mentioned in $Y$.
We denote the set of all (dual-)indeterminates occurring in a monomial $m$ by $\var(m) = \{ x \in X \cup \nnX \mid m(x) > 0 \}$.
By examining what combination of indeterminates from $X^\pm$ occur in the monomials of $\pirev^\pm \ext {\FormulaWin(v)}$, we can read off all minimal repairs as follows.

\begin{proposition}\label{reverseBothDirections}
In the above setting, the following holds:
\begin{enumerate}
\item 
Let $m \in \pirev^\pm \ext {\FormulaWin(v)}$ be a monomial.
Then the set $E(\var(m) \cap X^\pm)$ is a repair.

\item
Let $R \subseteq E^\pm$ be a repair.
Then there is a monomial $m \in \pirev^\pm \ext {\FormulaWin(v)}$ such that
$E(\var(m) \cap X^\pm) \subseteq R$.
If $R$ is minimal, then $E(\var(m) \cap X^\pm) = R$.
\end{enumerate}
\end{proposition}

\noindent
Before proving \Cref{reverseBothDirections}, we illustrate the computation of minimal repairs in a small example.

\begin{Example}
In the following game, Player~1 wins from $v$.
We are interested in the minimal repairs with $E^+ = \{c\}$ and $E^- = \{a,b\}$.
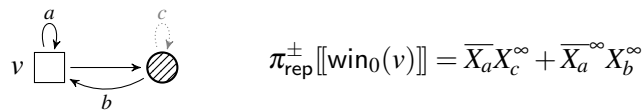
\begin{center}
\begin{tikzpicture}[node distance=1.5cm, baseline]
\node [p1,label={left:$v$}] (0) {};
\node [p0,F,right of=0] (1) {};
\draw [arr]
    (0) edge (1)
    (1) edge [bend left] node {$b$} (0)
    (1) edge [loop above,densely dotted,gray] node {$c$} (1)
    (0) edge [loop above] node {$a$} (0)
    ;
\end{tikzpicture}
\hspace{1cm}
$\displaystyle \pirev^\pm \ext {\FormulaWin(v)} = \nn{X_a} X_c^\infty + \nn{X_a}^\infty X_b^\infty$\\[.3cm]
\end{center}
Evaluating $\FormulaWin(v)$ in the $\Sinf[X,\nnX]$-interpretation $\pirev^\pm$ described above results in two monomials.
The first yields the repair $\{a,c\}$, the second yields the minimal repair $\{a\}$ (notice that $X_b \notin X^\pm$, as edge $b$ is already present).
The reason why we get two monomials is that we track also positive usage of edge $b$ by $X_b$, but are only interested in the negative indeterminate $\nn{X_b}$ for the repairs.
\end{Example}

\begin{proof}[Proof of \Cref{reverseBothDirections}]
We prove both statements by considering homomorphisms into the Boolean semiring $\Bool$.
For the first statement, let $m \in \pirev^\pm \ext {\FormulaWin(v)}$ be a monomial and
let $h \colon X \cup \nnX \to \Bool$ be the unique function that respects dual-indeterminates and satisfies
\begin{itemize}
\item $h(x) = \Bone$, for all $x \in \var(m)$
\item $h(X_e) = \Bzero$, if $X_e, \nn{X_e} \notin \var(m)$ and $e \in E^+$ (do not add $e$ without reason),
\item $h(X_e) = \Bone$, if $X_e, \nn{X_e} \notin \var(m)$ and $e \in E^-$ (do not remove $e$ without reason).
\end{itemize}

Then, $h$ lifts to a fully-continuous semiring homomorphism $h \colon \Sinf[X,\nnX] \to \Bool$ with $h(m) = \Bone$.
Moreover, $h \circ \pirev^\pm$ is a Boolean interpretation which corresponds to a Boolean model $\Gg'$.
Since semiring semantics are preserved by fully-continuous homomorphisms, we have $h \circ \pirev^\pm \ext {\FormulaWin(v)} = h(\pirev^\pm \ext {\FormulaWin(v)}) \ge h(m) = \Bone$ and hence $\Gg' \models \FormulaWin(v)$.
By the choice of $h$, the model $\Gg'$ is equal to $\Gg$ except that we add all edges $e \in X^+$ with $X_e \in \var(m)$, and remove all $e \in X^-$ with $\nn{X_e} \in \var(m)$.
Hence $\Gg'$ results from $\Gg$ by adding or deleting the edges $E(\var(m) \cap X^\pm)$, and since $\Gg' \models \FormulaWin(v)$, this set is a repair as claimed.

For the second statement, let $R \subseteq E^\pm$ be a repair and consider the repaired game $\Gg' \models \FormulaWin(v)$.
As $\Gg'$ differs from $\Gg$ only by edges in $E^\pm$, there is a unique assignment $h \colon X \cup \nnX \to \Bool$ such that $h \circ \pirev^\pm$ corresponds to $\Gg'$.
Again, $h$ lifts to a fully-continuous homomorphisms and we thus get $\Bone = h \circ \pirev^\pm \ext {\FormulaWin(v)} = h( \pirev^\pm \ext {\FormulaWin(v)})$.
So there must be a monomial $m \in \pirev^\pm \ext {\FormulaWin(v)}$ with $h(m) = \Bone$.
Consider the set $\var(m) \cap X^\pm$.
If $X_e \in \var(m) \cap X^\pm$, then $h(X_e) = \Bone$ and hence $e \in R$ by construction of $h$.
Further, $\nn{X_e} \in \var(m) \cap X^\pm$ implies $h(\nn{X_e}) = \Bone$ and thus again $e \in R$ by construction of $h$.
This proves $E(\var(m) \cap X^\pm) \subseteq R$.
If $R$ is minimal, we have equality:
otherwise $E(\var(m) \cap X^\pm)$ would be a smaller repair by the first statement, contradicting minimality.
\end{proof}

We remark that these results ignore the exponents of the monomials, so we could drop exponents from $\Sinf[X,\nnX]$ and work in the resulting, simpler semiring $\PosBool[X,\nnX]$ (but again, the number of minimal repairs can be exponential in the size of the game).
Further, the reverse analysis approach is not limited to questions about edges. We can also work with
interpretations that track the target set $F$ and thus answer questions such as \textit{how to choose or modify the target set so that Player $0$ wins?}

\section{Conclusion}

Based on a recent line of research on semiring provenance analysis that lead from database theory to semiring semantics for LFP,
we reported here on a case study that puts semiring semantics to use for a strategy analysis in Büchi games.
The choice of Büchi games has been motivated on one side by their relevance for applications in 
the synthesis and verification of reactive systems, on the other side because they provide one of the simplest non-trivial
cases of infinite games for which the definability of winning positions requires an alternation between least and greatest fixed points --
and can thus not be treated by simpler classes of semirings such as the $\omega$-continuous ones 
used for Datalog and reachability games. 

The aim of the case study was to illustrate how semiring semantics can be applied to more complex games,
featuring infinite plays and complicated winning conditions, and what kind of insights it provides (or fails to provide)
about the winning strategies in the game.
This is captured in the central Sum-of-Strategies Theorem and its applications.
This non-trivial result can be seen as a simpler version of the general Sum-of-Strategies characterization in terms of model-checking games in \cite{DannertGraNaaTan21}
and it essentially identifies the value of the statement that Player~0 wins with the sum of the
valuations of all (absorption-dominant) winning strategies.
While this applies to the class of all absorptive, fully-continuous semirings,
the most important semirings for our analysis are generalized absorptive polynomials $\Sinf[X]$.
Due to their universal property, these provide the most general information,
allowing us to read off the edge profiles of all absorption-dominant strategies.

With this information, we can count positional strategies,
we can determine whether a particular move is needed (once or even infinitely often) for winning the game,
and we can compute minimal ``repairs'' for a game.
The method of semiring valuations is rather flexible; we can use different semirings than $\Sinf[X]$, and we can tailor the set of moves that 
we track to make the resulting polynomial smaller and its computation more efficient.
We remark that, of course, there are also limitations to this method and not all relevant questions 
about strategies can be answered directly by semiring valuations.
As an example we show in the full version of this paper \cite{GraedelLucNaa21} 
that minimal cost computations provide serious obstacles to a semiring treatment.

The Sum-of-Strategies result motivates the general notion of absorption-dominant strategies
which captures strategies that are minimal with respect to the multiplicities of the edges they use.
To understand these strategies, we discussed how they relate to other classes of simple strategies, namely to positional and persistent strategies,
and we have shown that these form a strict hierarchy.

Finally, we remark that although Büchi games have been chosen as the topic of our case study,
the method of semiring valuations in absorptive semirings is not confined to this case. In principle,
it can be applied to different formulae and generalizes in particular to other games such as parity games,
as long as the winning positions are definable in fixed-point logic.
The win-formula for parity games is more complicated, and is parametrised by the number of priorities,
and so the Sum-of-Strategies Theorem requires different technical 
details, but can be established along the same lines.


\bibliographystyle{eptcs}
\bibliography{paper}

\end{document}